\documentclass[11pt]{article}
\usepackage[english]{babel}
\usepackage[utf8]{inputenc}
\usepackage[T1]{fontenc}
\usepackage{authblk}
\usepackage{graphicx}
\usepackage{csquotes}
\graphicspath{ {figures/} }
\usepackage{float}
\usepackage{tikz}
\tikzset{every picture/.style={line width=0.75pt}} 
\usetikzlibrary{patterns}
\usepackage{amsmath,mathrsfs,dsfont}
\usepackage{amsthm}
\usepackage[all]{xy}
\usepackage{graphics}
\usepackage{dirtytalk}
\usepackage{tensor}
\usepackage{verbatim}
\usepackage{nicefrac}
\usepackage{amsfonts}
\usepackage{amssymb,latexsym}
\usepackage{color}
\usepackage{mathrsfs,dsfont}
\usepackage{physics}
\usepackage{mathtools}
\usepackage{bbm}
\usepackage{appendix}
\usepackage{slashed}
\usepackage{cite}
\usepackage{url}
\usepackage{hyperref}
\usepackage{comment}
\newtheorem{theorem}{Theorem}[]
\newtheorem{corollary}[theorem]{Corollary}
\newtheorem{lemma}[theorem]{Lemma}
\newtheorem{remark}[theorem]{Remark}

\newtheorem{proposition}{Proposition}[section]

\numberwithin{equation}{section}

\newcommand{\ch}{{\large\bf CHANGE !!!  }}
\newcommand{\nn}{\nonumber}
\newcommand{\dis}{\displaystyle}
\newcommand{\bea}{\begin{eqnarray}}
\newcommand{\eea}{\end{eqnarray}}
\newcommand{\supp}{\operatorname{supp}}
\newcommand{\dist}{\operatorname{dist}}

\begin{document}

\title{A Nontrivial Lattice Scalar Quantum Field Model}
\author[a]{Paulo A. Faria da Veiga\thanks{veiga@icmc.usp.br}}
\author[a]{Ravi Mistry}
\author[a]{Michael O'Carroll\thanks{\em In Memorian.}}
\affil[a]{Instituto de Ciências Matemáticas e de Computação (ICMC), 
Universidade de São Paulo,
Avenida Trabalhador São-carlense 400, 13566-590 São Carlos SP, Brazil}

\date{\today}

\maketitle

\begin{abstract}
We introduce a new Euclidean quantum model with zero-spin scaled scalar fields on a spacetime hypercubic lattice, which we refer to as the hybrid model. The model coupling parameters are the hopping parameter $\kappa^2>0$ and the field interaction coupling $\alpha>0$. Using the technique of lattice scaled fields, with a noncanonical scaling that corresponds to a partial renormalization (wave function renormalization), we show that the partition function of this model satisfies both thermodynamic and ultraviolet stability bounds, with constants which are uniform in the lattice volume $\Lambda_s$ and the lattice spacing $a\in(0,1]$. As a consequence, the model has a finite free energy per degree of freedom in the thermodynamic limit $\Lambda_s\nearrow\infty$, at least in the subsequential sense. Subsequently, we can also take the ultraviolet or continuum limit $a\searrow 0$. Besides the finiteness property of the free energy, for sufficiently small values of the model coupling parameters $\kappa^2$ and $\alpha$, we show the model satisfies Osterwalder-Schrader positivity and construct the underlying physical quantum mechanical Hilbert space, and, using a polymer expansion, we show correlations also remain bounded in the thermodynamic limit. Finally, we perturbatively compute an infinite volume truncated four-point correlation showing it is nonzero. Noting that, for fixed $a$, perturbation series in $\kappa^2$ and $\alpha$ converges, these results are sufficient to show that our hybrid model makes sense as a lattice quantum field theory in spacetime dimension $d=3,4$, with possible extensions to other higher values of $d$.
\end{abstract}

\tableofcontents

\vspace{8mm}

\noindent \ch  No {\bf Hybrid} in TITLE. No ONE KNOWS THIS NAME AT tHE BEGINNING !

\section{\label{sec:level1} Introduction and the Model Definition}

Although significant progress has been made, we have never been able to prove the compatibility between special relativity and quantum mechanics in a system with an infinite degrees of freedom. There is no rigorous proof of existence for a nontrivial quantum field model in Euclidean spacetime that exists in both the thermodynamic and ultraviolet (UV) limits and verifies the Osterwalder-Schrader-Seiler axioms \cite{Osterwalder:1973dx,Osterwalder:1974tc,Osterwalder1978gauge,Glimm:1987ng,Seiler:1982pw}. Recently, the quest for nontriviality has gained even greater importance, as the scalar $\lambda \phi^4$ model was shown to be trivial and non-interacting (Gaussian) for $\lambda > 0$ in $d$ Euclidean dimensions, using probabilistic methods \cite{Aizenman:2019yuo}.

In this context, using the formalism of scaled fields which incorporates a partial renormalization in the UV through a noncanonical scaling (a sort of an a priori wavefunction renormalization) \cite{OCarroll:2018uom}, we introduce a new interacting scaled toy model for the real scalar field $\phi$, which we refer to as the \emph{hybrid model}. The model is defined on a finite $d$-dimensional Euclidean hypercubic lattice $\Lambda \subset a\mathbb{Z}^d$, with the Euclidean spacetime dimension $d$, a lattice spacing $a \in (0,1]$ and $\Lambda_s$ sites. We denote its partition function by $Z_\Lambda(\kappa^2,\alpha)$. For each lattice site $x \in \Lambda$, we assign a field variable $\phi_x \equiv \phi(x)$. The forward-translated site in the $\mu$ direction is denoted by $x_\mu^+$ and is defined as $x_\mu^+ \doteq x + a\, e^\mu$, where $e^\mu$ is the unit vector in the $\mu$-th spacetime direction. The label $\mu = 0$ denotes the Euclidean time direction.

The partition function of the hybrid model on $\Lambda$ takes the form
\bea\label{hybrid-partition}
Z_{\Lambda}(\kappa^{2},\alpha)
=\int d\mu(\phi)\, B_\Lambda(\phi)\,,
\eea
where the hybrid model Boltzmann factor $B_\Lambda(\phi)$ is given by
\bea\label{action}
B_\Lambda(\phi)\,=\exp\!\Big(\kappa^{2}\sum_{x,\mu}\phi_{x}\phi_{x_{\mu}^{+}}\Big)
\prod_{x,\mu}\big(1+\alpha\,\phi_{x}\phi_{x_{\mu}^{+}}\big),
\eea
with a hopping parameter $\kappa^2>0$, and the sum and the product running over $x = (x^0, x^1,\ldots, x^{d-1})\in\Lambda$ and $\mu=0,1,\dots,d-1$. The measure appearing in \eqref{hybrid-partition} is the normalized Gaussian product measure
\bea \label{measure}
d\mu(\phi)=\prod_{x\in\Lambda}\frac{d\phi_x}{\sqrt{2\pi}}\,e^{-\frac{1}{2}\phi_x^2},
\eea
of Gaussian measures with mean–zero and unit covariance.

For $n=0,1,2,\ldots$, the $2n$-point correlations are defined by by the statistical averages
\bea \label{2n-pt.funct.}
G_\Lambda(x_1,\ldots,x_{2n})\doteq\frac{1}{Z_\Lambda}\;\int d\mu(\phi)\quad \phi_{x_1}\ldots\phi_{x_{2n}}\,B_\Lambda(\phi)\,.
\eea
For an odd number of external fields, the correlations vanish by the symmetry $\phi\rightarrow (-\phi)$.

We note that the nonlocal $\alpha$-polynomial interacting terms, with a coupling $\alpha\geq 0$, correspond to an $\alpha$ Taylor expansion truncation of the exponential term
\bea
\exp\!\Big[\alpha\,\sum_{x,\mu} \phi_x\,\phi_{x^+_\mu}\Big]\,.
\eea
When the full exponential is retained, the model becomes a scaled free scalar model with a modified hopping parameter $(\kappa^2+\alpha)>0$ rather than $\kappa^2>0$. The (exponential) scaled free field model partition function gives an upper bound on the hybrid model partition function. Setting $\alpha=0$, we obtain the usual scaled free field model with hopping parameter $\kappa^2$ which is well understood (see, e.g., appendix~A in \cite{daVeiga:2022twu} for a detailed analysis). 

From this, considering the hybrid model \eqref{hybrid-partition} in generic spacetime dimension $d\geq 2$, one must have $(\kappa^2+\alpha)\in(0,1/2d]$, and the value $(\kappa^2+\alpha)=1/2d$ describes the massless, critical model case. This is the content of Lemma \ref{lemma_TUV} below. Specifically, for $\alpha>0$, we obtain a triangular domain around the origin in the first quadrant of the $\kappa^2$--$\alpha$ plane, with sides of length $1/2d$.

Now, as we argue below, when one considers, instead of the exponential interaction, the polynomial $\alpha$-interaction as in \eqref{action}, some of our results extend to an infinite horizontal rectangular strip of height $1/2d$ and limited below by the horizontal $\alpha$ axis. In particular, this infinite domain holds when proving the Thermodynamic and Ultraviolet stable stability bounds or, simply, TUV bounds on the partition function $Z_\Lambda(\kappa^2,\alpha)$. 

In the sequel, we first prove that the hybrid model partition function $Z_\Lambda(\kappa^2,\alpha)$ satisfies the TUV bounds. Namely, there exist two real constants $C_\ell$ and $C_u$, independent of the {\em the lattice volume} $\Lambda_s$ and the lattice spacing $a$, such that
\bea \label{TUV}
\exp\big[C_\ell \Lambda_s\big] \leq Z_\Lambda(\kappa^2,\alpha) \leq \exp\big[C_u \Lambda_s\big].
\eea
The TUV bounds in \eqref{TUV} leads, via the Bolzano-Weierstrass theorem (see e.g. thms. 2.41 and 2.42 in \cite{Rudin1976principles}), to a normalized free energy per degree of freedom
\bea \label{FE}
F_\Lambda(\kappa^2,\alpha) \equiv \frac{1}{\Lambda_s} \ln Z_\Lambda(\kappa^2,\alpha)\,,
\eea
which exists in the thermodynamic limit as $\Lambda_s \nearrow \infty$ and in the continuum limit as $a \searrow 0$, at least in the sense of subsequences. The finiteness of the free energy in the thermodynamic and UV limits explains why we call the stability bound of (\ref{TUV}) as thermodynamic and ultraviolet stable stability bounds or, simply, TUV bounds.
\begin{remark}
Note that, in order to show the existence of the normalized free energy in the thermodynamic and continuum limits, in the TUV bounds of  \ref{TUV}, it is fundamental same type of dependence on the {\em volume} $\Lambda_s$ of $\Lambda$, as well as to have the constants $C_\ell$ and $C_u$ independent both on $\Lambda_s$ and then lattice spacing $a$. Although the existence and finiteness of (the limits being understood in the subsequential sense)	$$F(\kappa^2,\alpha)\,\equiv\,\lim_{a\searrow 0}\,\;\left[\lim_{\Lambda_s\nearrow \infty}\;F_\Lambda(\kappa^2,\alpha)\right]$$ 
does not characterize the limiting model, it does prove an important finiteness property for the model.
\end{remark}

Throughout this paper, for simplicity of the analysis, we take the {\em small parameter domain} 
\bea \label{small domain}
0 <\kappa^2<\kappa_0^2 \ll 1\quad\mathrm{and}\quad 0<\alpha<\alpha_0 \ll 1\,, 
\eea
for some constants $\kappa_0,\alpha_0>0$. Note that, within the small-parameter domain \eqref{small domain}, no large-field problem occurs. Integrability of the fields at infinity is ensured by the local Gaussian measure \eqref{measure}.

The results presented below hold assume the small parameter domain in \ref{small domain} and most of our results extend to spacetime dimension $d \geq 3$. Although our primary interest is in spacetime dimension $d=4$, several statements are therefore formulated for general $d \geq 3$. (For $d=2$, as in the continuum spacetime formulation, difficulties arise due to singularities of the Laplacian resolvent at coincident points.)

We remark that, for sufficiently small $\kappa^2>0$ and $\alpha>0$, we can employ perturbation-like methods, which are more easily controlled, and we can construct e.g. convergent polymer expansions for the model correlations, following \cite{Glimm:1987ng, Simon1993}. It must be emphasized, however, that with additional effort our results here can be extended from the small parameter domain to larger regions of the  $(\kappa^2,\alpha)$ plane.

We conclude the description of the model by noting that two of us developed the lattice scaled field method in \cite{FariadaVeiga:2019cxp,OCarroll:2018uom}, applicable to both bosonic and fermionic Euclidean field models. For real-valued scalar boson fields, the relation between the scaled and unscaled (physical, original) fields, denoted with an upper index $u$, is
\bea \label{scaling}
\phi(x)=a^{(d-2)/2}\,\left(m_u^2a^2+2d\kappa_u^2\right)^{1/2}\,\phi^u(x).
\eea
Here $m_u\ge 0$ and $\kappa_u^2>0$ are the unscaled, physical bare mass and hopping parameter, respectively. As obtained in \cite{daVeiga:2022twu}, the relation between the scaled parameter $\kappa$ of \eqref{action} and the unscaled, original hopping parameter $\kappa_u$ is given by
\bea\label{kappau}
\kappa^2\,=\,\left[2d\,+\,\left( \dfrac{m_ua}{\kappa_u} \right)^2 \right]^{-1}\,.
\eea
Roughly speaking, our non canonical scaling is gives a partial renormalization that smooths ultraviolet singularities and yields to greater regularity for the model in the continuum limit $a\searrow 0$.

The scaled field method has also been successfully applied more recently to lattice Yang--Mills models constructed with the Wilson plaquette action \cite{daVeiga:2022twu}. Again, although the method does not resolve the fundamental question of the existence and physical properties of the limiting models in the thermodynamic ($\Lambda_s\nearrow\infty$) and continuum or UV ($a\searrow 0$) limits, it does allow us to prove important results on the finiteness of many relevant quantities of the model. For instance, correlations are finite for the scaled free model. Moreover, plaquette field correlations are finite in $d=4$ Yang-Mills theory, and the two-point correlation of the scaled free model remains finite even at coincident points for $d\geq 3$.

The paper is organized as follows. In Section~\ref{sec2}, we present our main results, formulated in several lemmas and theorems. Section~\ref{sec3} contains the proof of the TUV bounds (Lemma~\ref{lemma_TUV}) and the finiteness of correlations, Theorem \ref{thm. 2-pt. funct.} and its corollary. TUV bounds lead to a finite scaled free energy per degree of freedom. Together with the finiteness of correlations shown, we have a warranty of a proper mathematical existence of the model. Next, in Section~\ref{sec4}, we compute a truncated four-point function and show it is nonzero. This analysis needs the knowledge of the two-point correlation. These results are established in Theorem \ref{thm. 4-pt. funct.}. From these results, we conclude that the hybrid model fulfill the requirements to be the first known example of a both mathematically well-defined and nontrivial quantum field theory model. Finally, Section~\ref{sec5} provides our concluding remarks.

\section*{Acknowledgments}
PAFV and RM thank very much the late friend and colleague Michael O´Carroll for collaboration on this work. PAFV acknowledges and thanks Michael O´Carroll for a long standing collaboration. The work of RM was supported by the CNPq PDJ grant 168581/2023-0. RM also thanks the Institute of Physics of São Carlos (IFSC) and the Institute of Mathematical Sciences and Statistics (ICMC), University of São Paulo at São Carlos, for hospitality during the preparation of this work.

\section{Preliminaries and Results} \label{sec2}

\subsection{The Physical Hilbert Space} \label{sec2.1}
The standard Osterwalder--Schrader (O-S) axioms \cite{Osterwalder:1973dx, Osterwalder:1974tc, Glaser:1974hy} for the scalar fields are the foundations of any well-defined (Euclidean) quantum field theory model. The Osterwalder--Seiler extension \cite{Osterwalder1978gauge, Seiler:1982pw} to include fermions is also a crucial component of these foundations. This requirement, commonly referred to as reflection positivity (see \cite{Neeb20018, Seiler:1982pw}), is formulated through the antilinear time-reflection operator $\Theta$, which reverses the Euclidean time coordinate, $x^{0}\mapsto -x^{0}$, while taking the complex conjugate of the functional values. In particular, $\Theta$ maps each field configuration to its time-reflected image, ensuring that the reflected product is nonnegative for all suitable functionals $F$ supported in the positive-time region. Although establishing this property is often the most delicate step, in the present hybrid model, this property follows naturally. Moreover, for $\kappa^{2}>0$ and $\alpha>0$, and along with Griffiths I inequality for ferromagnetic-like spin systems (see, e.g., \cite{Glimm:1987ng}), the hybrid Boltzmann factor \eqref{action} and the normalized Gaussian measure \eqref{measure} are both invariant under time reflection. As a result, expectation values constructed with $\Theta$ satisfy the desired positivity condition. Once this property is established, the physical (separable) Hilbert space associated with the Euclidean lattice theory can be constructed. In the end, if all the other axioms are verified (see also \cite{Glimm:1987ng}), in the thermodynamics limit and the continuum limit, one obtains a corresponding hybrid model in Minkowski spacetime that satisfies the Wightman axioms \cite{Streater2000pct,Jost:1965}. It is undeniable that the continuum limit remains an open challenge; however, with our scaled fields~\eqref{scaling}, the limiting behaviors appear to improve, as will be discussed later, see also \cite{daVeiga:2022twu, FariadaVeiga:2019cxp,OCarroll:2018uom}.

Let $\Lambda\subset\mathbb{Z}^d$ be a finite hypercubic lattice. Our aim here is to describe as simply as possible the construction of the underlying physical Hilbert space $\mathcal{H}$ of the model. To this end, we first construct a pre-Hilbert space carrying a positive definite inner product; its completion will then yield $\mathcal{H}$.


Assume that $\Lambda$ is symmetric with respect to the hyperplane $x^0=0$, and consider the decomposition
\bea
\Lambda=\Lambda_+\cup\Lambda_-,\qquad \Lambda_+\cap\Lambda_-=\emptyset,
\eea
where $\Lambda_+$ (respectively $\Lambda_-$) denotes the \emph{positive-time} (respectively \emph{negative-time}) part of the lattice. Let $\mathcal{A}_\Lambda$ be the \emph{complex involutive algebra} generated by the family of scalar fields $\{\phi_x\}_{x\in\Lambda}$; that is, $\mathcal{A}_\Lambda$ consists of all finite polynomials in the field variables with complex coefficients.\footnote{One may regard $\mathcal{A}_\Lambda$ as an involutive algebra over the involutive base ring $\mathbb{C}$. The involution is taken as part of the algebraic structure. Specifically, it is a conjugate linear, anti multiplicative map whose restriction to the scalars coincides with complex conjugation. This contrasts with operator algebraic approaches, where the involution is concretely realized as the adjoint of operators acting on a Hilbert space. For example, an involutive algebra of bounded operators on a Hilbert space admits a norm closure which is a $C^*$-algebra; moreover, if it is closed in the weak (or strong) operator topology, it is a von Neumann algebra.} The expectation functional
\bea
\expval{\cdot}_\Lambda:\mathcal{A}_\Lambda\to\mathbb{R}
\eea
is defined by the integral measure introduced in \eqref{measure}, i.e. by integration against the normalized Gaussian product measure weighted with the Boltzmann factor \eqref{action}, and is assumed to be reflection invariant.
Specifically, let $r:\Lambda\to\Lambda$ denote the reflection at the hyperplane $x^0=0$, so that $r(\Lambda_+)=\Lambda_-$ and $r(\Lambda_-)=\Lambda_+$. This induces an antilinear time-reflection operator
\bea
\Theta:\mathcal{A}_\Lambda\to \mathcal{A}_\Lambda
\eea
defined by
\begin{enumerate}
  \item $\Theta(\phi_x)=\phi_{r x}$ for all $x\in\Lambda$;
  \item $\Theta(FG)=\Theta(F)\Theta(G)$ for all $F,G\in \mathcal{A}_\Lambda$;
  \item $\Theta(\lambda F)=\bar{\lambda}\,\Theta(F)$ for all $\lambda\in\mathbb{C}$, $F\in \mathcal{A}_\Lambda$,
\end{enumerate}
and, in particular, $\Theta(\mathcal{A}_{\Lambda_+})\subset \mathcal{A}_{\Lambda_-}$, where $\mathcal{A}_{\Lambda_\pm}$ denotes the subalgebra generated by $\{\phi_x \mid x\in\Lambda_\pm\}$.

\begin{theorem}[Reflection positivity] \label{O-S positivity}
For every $F\in \mathcal{A}_{\Lambda_+}$ one has
\bea
\expval{\Theta F \cdot F}_\Lambda \;\geq\; 0.
\eea
\end{theorem}
\begin{proof}
Although technically involved, the proof follows essentially the same line of reasoning as in \cite{Seiler:1982pw}, adapted to the present scaled scalar field setting; see also \cite{Osterwalder:1973dx,Osterwalder:1974tc,Glaser:1974hy}.
\end{proof}

This reflection positivity property allows us to construct the quantum mechanical \emph{complex Hilbert space} $\mathcal{H}$. Define the null space
\bea
\mathcal{K}\doteq\Big\{F\in \mathcal{A}_{\Lambda_+}\,\Big\lvert\,\langle \Theta F\cdot F\rangle_\Lambda=0\Big\}.
\eea
Then the quotient space $\mathcal{A}_{\Lambda_+}/\mathcal{K}$ becomes a pre-Hilbert space with inner product
\bea
\big([F],[G]\big)_{\mathcal{H}}\doteq \langle \Theta F\cdot G\rangle_\Lambda,\qquad F,G\in \mathcal{A}_{\Lambda_+},
\eea
where $[F]$ denotes the equivalence class of $F$ in $\mathcal{A}_{\Lambda_+}/\mathcal{K}$. This is the usual Feynman–Kac type formula, cf. \cite{Glimm:1987ng,RS1975ii}. In particular, in terms of the scalar fields one has
\bea \label{H-inner prod.}
\big([\phi_x],[\phi_y]\big)_{\mathcal{H}} = \langle \phi_{r x}\,\phi_y\rangle_\Lambda,\qquad x,y\in\Lambda_+,
\eea
Finally, the completion of this pre-Hilbert space with respect to $(\cdot,\cdot)_{\mathcal{H}}$,
\bea
\mathcal{H}\doteq\big[\mathcal{A}_{\Lambda_+}/\mathcal{K}\big]_c,
\eea
is what we call the \emph{physical Hilbert space} of the model. The existence of this Hilbert space, together with the O-S axioms, enables the construction of a well-defined Minkowski spacetime model in the Wightman framework. This is the underlying motivation for the results presented below.

\subsection{Main results}

Although we are mainly interested in $d=4$, we emphasize that several of the results below generalize to broader values of the spacetime dimension $d$ and for larger admissible ranges of $\kappa^2$ and $\alpha$, encompassing the small-parameter domain \eqref{small domain}.

Let us denote the lattice bond $b$, between the sites $x$ and $x_{\mu}^{+}$,  by $b=(x,x_{\mu}^{+})$ and define
\bea \label{def. Xb}
X_b(\phi) \doteq \phi_x\phi_{x_{\mu}^{+}}.
\eea
Recalling $\Lambda_s$ denotes the number of lattice sites in the lattice $\Lambda$, we denote by $\Lambda_b$ the number of bonds in $\Lambda$. We also assume, for $\kappa^2,\alpha>0$, that
\bea \label{1/2d}
0<\kappa^{2}+\alpha<\frac{1}{2d}.
\eea
Let $\Delta$ denote the nonnegative {\em minus} lattice Laplacian operator.  We adopt the usual periodic boundary conditions on $\Lambda$ \cite{daVeiga:2022twu}. In this case, the $\Delta$ and other related operators are well defined (see below).

\begin{lemma} \label{lemma_TUV}
Let $d=3,4,\ldots$ denote the lattice spacetime dimension and $\mathcal{E}>0$. Define
\bea \label{def. S_E}
S_{\mathcal{E}}\doteq\Bigl\{\phi:\Lambda\to\mathbb{R}\,\,\Bigl\lvert\,\,|X_b|\le\mathcal{E}\,\,\forall\, b\Bigr\}.
\eea
Then using periodic boundary conditions on $\Lambda$ and with the notation given above, one has the two-sided TUV stability bound for the partition function  $Z_\Lambda(\kappa^{2},\alpha)$ given in \eqref{hybrid-partition}
\bea\label{Z-bounds}
e^{C_\ell\Lambda_{s}}
\;\le\; Z_{\Lambda}(\kappa^{2},\alpha)
\;\le\; e^{C_u\Lambda_{s}},
\eea
where the real constants are\footnote{Here, the operator is understood in terms of its spectral representation.}
\bea
C_\ell\doteq -d\kappa^{2}\mathcal{E}, \qquad C_u\doteq  \frac{1}{2\left(\kappa^{2}+\alpha\right)}\left(\Delta+r\right)^{-1} (0,0),
\eea
and
\bea\label{var. r}
r=r(\kappa^2,\alpha)\doteq(\kappa^{2}+\alpha)^{-1}-2d\,.
\eea
\end{lemma}
\begin{remark}
The TUV stability bound is valid in the extended domain of the parameters corresponding to the infinite rectangle in the $(\alpha,\kappa^2)$ plane, given by  $\kappa^2\in[0,1/2d]$ and $\alpha>0$. For $d=2$ we must restrict $\kappa^2$ to $[0,1/2d)$, by the resons given above.
\end{remark}
The proof of this lemma, together with the result on the finiteness of the free energy \eqref{FE}, will be given in Section~\ref{sec3}. We next state a result on the non-vanishing of the truncated two-point function $G_\Lambda(x,y)$ and on the characterization of the critical line in the triangular domain described above.

\begin{remark}
Since the underlying Gaussian measure~\eqref{measure} is centered at zero, it is invariant under $\phi \mapsto -\phi$, and therefore all correlations involving an odd number of fields do vanish. In particular, the one-point function is identically zero. Hence the truncated two-point function coincides with the non-truncated one throughout our analysis.
\end{remark}
\begin{theorem} \label{thm. 2-pt. funct.}
Let $\kappa^{2},\alpha>0$, and let $Z_\Lambda(\kappa^{2},\alpha)$ be the partition function of \eqref{hybrid-partition}. Then, the (truncated) two–point function $G_\Lambda(x,y)$, evaluated up to order $\mathcal{O}(\kappa^{4})$ and $\mathcal{O}(\alpha^{2})$, is non–vanishing. Its Fourier transform at zero momentum, and hence the convolution inverse of $G_\Lambda(x,y)$, yields the  critical line (massless case)
\bea \label{critical}
\kappa^{2}=\frac{1}{2d}+\alpha^{2}-\alpha .
\eea
\end{theorem}
\begin{remark}
Using the Payley-Wiener theorem in the spectral representation for the two-point correlation $G_\Lambda(x,y)$, evaluated up to order $\mathcal{O}(\kappa^{4})$ and $\mathcal{O}(\alpha^{2})$, is non–vanishing. Its Four, the above theorem guarantees that, in the small parameter $\kappa^2$ and $\alpha$ domain we assume in this paper, that we are away from the critical line, the model has a nonzero mass gap, and the two-point correlation $G_\Lambda(x,y)$ exhibits an exponential distance decay at non coincident points. This is a key point to show convergence of, for example, perturbation-like expansions, as polymer expansions (see e.g. \cite{Glimm:1987ng}).
\end{remark}

We now establish our finiteness results for the thermodynamic limit of correlations of the hybrid model.
\begin{theorem} \label{thm. correl. funct.}
	Let $a\in(0,1]$ and $\kappa^{2},\alpha>0$ lie away from the critical line \eqref{critical}, as is the case of the small parameter domain. Then the hybrid model correlations exist and are finite in the thermodynamic limit $\Lambda\nearrow a\mathbb Z^d$. Furthermore, correlations exhibit exponential tree decay. 
\end{theorem}
\ch {\bf NEW the infinite volume is needed below, otherwise your computation is not correct and periodic boundary condition must be taken into account.}

\begin{remark}
From now on we assume to be in the thermodynamic limit and denote the infinite-volume correlations without the index $\Lambda$.
\end{remark}

Once the results on the two-point functions are established, we obtain the following non-vanishing result for the truncated four-point correlation of the model.

\ch {\bf why not kappasquare and alpha in the small parameter domain????? in principle, it can be chosen smaller than 2.15. but i see no reason for 2.15 to be satisfied. I prefer the small parameter domain since it gives the chance for a convergent Taylor expansion in kappasquare  and alpha}

\ch {\bf more precision: the 4 and 2 point fct are in the infinite volume here!}

\begin{theorem} \label{thm. 4-pt. funct.}
Let $\kappa^{2},\,\alpha>0>0$ be \ch in the small parameter domain
and let $Z_\Lambda(\kappa^{2},\alpha)$ be the partition function defined in \eqref{hybrid-partition}. 
Define the truncated (connected) four--point function by
\bea \label{tr-4pt}
G^{\mathrm{T}}(x_1,x_2,x_3,x_4)\doteq G(x_1,x_2,x_3,x_4)
-\sum_{\text{pairings}}
G(x_i,x_j) G_\Lambda(x_k,x_l)\,,
\eea
where the sum runs over all pairwise contractions of $\{x_1,x_2,x_3,x_4\}$.  
Then \ch at least for the special configuration of points where $x_1\not= x_2\not= x_3 \not= x_4$, it follows that $G^{\mathrm{T}}(x_1,x_2,x_3,x_4)$ has a nonzero contribution at up to the joint perturbation order 
$\mathcal{O}(\kappa^{4})$ and $\mathcal{O}(\alpha^{2})$.
\end{theorem}
\ch {\bf !!!! NEW NEW}
\begin{remark}
We note that both the hopping term and the $\alpha$ interaction in our hybrid model Boltzmann factor are quadratic in the fields. Since, for fixed lattice $\Lambda$ and $a\in(0,1]$ we have 
$$
\Big| \dis\int\,d\mu(\phi)\:  [\phi(x)]^{2n}\,\Big| \;\leq\; C^n\,n! \,,   
$$
for some positive constant $C>0$ depending eventually on the lattice spacing $a\in(0,1]$. Thus, a joint perturbation (Taylor) series in $\kappa^2$ and $\alpha$ about the origin converges for small enough values of the couplings $\kappa^2$ and $\alpha$.
\end{remark}

\begin{remark}
The fact that $G^{\mathrm{T}}(x_1,x_2,x_3,x_4)\,\not=\,0$, at least for distinct points, accounts already for the nontriviality of the hybrid model in the small parameter domain. Above, we adopt the names pairing and contraction according to the standard Wick's theorem in quantum field theory.
\end{remark}

Taking into consideration the results of the above theorems ensure the hybrid model is a nontrivial quantum field theory model. 

\section{Proofs of the Main Results and Their Implications}\label{sec3}
\ch {\bf we need to have r>0. otherwise the operator B is not strictly positive! Due to this I am not sure anymore if Lemma 2 works for $\kappa^2$ and $\alpha$ not in the small parameter domeina}

\subsection{TUV Bounds and Existence of the Free Energy}

\begin{proof}[Proof of Lemma \ref{lemma_TUV}]
\emph{Upper bound.}
Since $1+u\le e^{u}$ for every $u\in\mathbb{R}$, we immediately obtain
\bea \label{exp_UB}
\prod_{b}\big(1+\alpha X_b\big)\leq\exp\!\Big(\alpha X_b\Big).
\eea
Therefore, using Griffiths I inequality \cite{Glimm:1987ng}   \ch {\bf I do not understand how the quadratic term in the rhs arises since you kept the gaussian measure intact. probably must replace $d\mu$ by $d\phi$ in last equality}
\bea\label{Z-ineq-1}
Z_{\Lambda}
&\le \dis\int d\mu(\phi)\,\exp\!\Big((\kappa^{2}+\alpha)\sum_{b}X_{b}\Big)
= \int d\mu(\phi)\,\exp\!\Big((\kappa^{2}+\alpha)\sum_{b}X_{b}-\tfrac{1}{2}\sum_{x}\phi_x^2\Big).
\eea
Using the periodic boundary conditions one rewrites the quadratic form as
\bea \label{op. B}
(\kappa^{2}+\alpha)\sum_{b}X_{b}-\tfrac{1}{2}\sum_{x}\phi_x^2
=-\tfrac{1}{2}(\phi,B\phi),
\qquad B=(\kappa^{2}+\alpha)(\Delta+r),
\eea
with $r>0$ as in \eqref{var. r}. We recall that $\Delta$ stands for minus the Laplacian, Hence the integral in \eqref{Z-ineq-1} is Gaussian and equals
\bea
Z_{\Lambda}\le \frac{1}{\det^{1/2}B}= \exp\!\Big\{\tfrac{1}{2}\sum_{j}\ln\lambda_j^{-1}\Big\},
\eea
\ch where $\{\lambda_j\}_{j\geq 1}>0$ since $B>0$. Using the elementary inequality $\ln u < u$, for all $u>0$, with $u = \lambda_j^{-1}$, we obtain  
\bea \label{trace formula}
\sum_j \ln \lambda_j^{-1} \le \sum_j \lambda_j^{-1} 
= \mathrm{Tr}\,B^{-1}
= \frac{\Lambda_{s}}{2\left(\kappa^{2}+\alpha\right)}\left(\Delta+r\right)^{-1}(0,0).
\eea
This yields the stated upper bound in \eqref{Z-bounds}. For $r>0$ (see \eqref{var. r} and \eqref{1/2d}), the operator $(\Delta + r)$ in \eqref{trace formula} admits a bounded inverse, i.e., its resolvent at coincident points is well-defined (see, e.g., \cite{daVeiga:2022twu}). While $Z_\Lambda$ remains finite for $r>0$ and $\kappa^{2} + \alpha < 1/(2d)$, in the limit $r \to 0$ (i.e., $\kappa^{2} + \alpha = 1/2d$) it is necessary to require $d \geq 3$ (cf. infrared divergences \footnote{See, for example, \cite{Derezinski2005} for a construction of the Hilbert space and the algebra of massless scalar field operators in $d=2$ dimensions, in the continuum spacetime. By working with both separable and non-separable Hilbert spaces, the authors obtain a formulation that avoids the usual infrared divergence, but does not satisfy all the Wightman axioms \cite{Streater2000pct,Jost:1965}. In particular, in the separable case, there does not exist a vacuum vector, while in the non-separable case, there are no (smeared) fields; but only the exponentiated (Weyl) operators \cite{Bratteli1997operator}.}).

\emph{Lower bound.}
Let $\mathcal{E}>0$ and restrict the integration to the set $S_{\mathcal{E}}$ defined in \eqref{def. S_E}. On $S_{\mathcal{E}}$ we have $X_b\ge -\mathcal{E}$, for every bond $b$. Hence, the following inequality holds
\bea
\exp\!\Big(\kappa^{2}\sum_{b}X_{b}\Big)\ge \exp(-\kappa^{2}\Lambda_{b}\mathcal{E})\,.
\eea
Moreover, for $\alpha>0$, the polynomial interaction factor $\prod_b(1+\alpha X_b)$ satisfies
\bea
\prod_b(1+\alpha X_b)\ge \prod_b(1-\alpha|X_b|)\ge 1-\alpha\sum_b|X_b|\,.
\eea
The latter inequality follows from repeated use of $(1-u_1)\cdots(1-u_n)\geq 1-(u_1+\cdots+u_n)$, for $u_i\geq 0$, $i=1,\ldots,n$.

Combining these estimates gives
\bea \label{l-bound}
Z_{\Lambda}
&\geq & \int_{S_{\mathcal{E}}} d\mu(\phi)\, \exp\!\Big(\kappa^{2}\sum_b X_b\Big)\prod_b\Big(1+\alpha X_b\Big)\nn\\
&\geq & \exp\!\Big(-\kappa^{2}d\Lambda_{s}\mathcal{E}\Big)\int_{S_{\mathcal{E}}} d\mu(\phi)\,\Big(1-\alpha\sum_b|X_b|\Big)\nn\\
&=&\exp\!\Big(-\kappa^{2}d\Lambda_{s}\mathcal{E}\Big)
,
\eea
which is the lower bound claimed in \eqref{Z-bounds}. (Needless to say, but for completeness, note that the integral in the last equality simplifies to 1. This is because the integral over the terms involving the bond vanishes, since $X_b$ (see \eqref{def. Xb}) contains only nearest-neighbor points, while for the rest the normalized integral measure gives 1.)
\end{proof}
\begin{remark}
In the proof of Lemma~\ref{lemma_TUV}, while establishing the lower bound \eqref{l-bound}, one could in principle improve it and find a better bound as follows. Since $d\mu(\phi)$ \eqref{measure} is the \emph{product of mean zero Gaussian measures}, the lower bound admits a more concrete combinatorial interpretation that automatically produces $Z_\Lambda\ge1$ without any additional condition. Expanding both the exponential and the nearest neighbor (nn) product in
power series,
\bea
\exp\:\Bigl(\kappa^{2}\:\sum_{b}X_b\Bigr)
   =\sum_{n\geq 0}\frac{\kappa^{2n}}{n!}
     \Biggl(\sum_{b}X_b\Biggr)^n,
\eea
and
\bea
\prod_{b}\bigl(1+\alpha\,X_b\bigr)
  =\sum_{S\subset\mathrm{nn\,pairs}}
     \alpha^{|S|}
     \prod_{b\in S}X_b,
\eea
cf. \eqref{prod. aX_b}, and integrating term by term (since all Gaussian moments are finite, for finite $\Lambda$ and fixed spacing $a\in(0,1]$), one obtains
\bea \label{Z-expansion}
Z_\Lambda
   =\sum_{n\geq 0}\;\sum_{S}
     \frac{\kappa^{2n}\alpha^{|S|}}{n!}
     \int d\mu(\phi)\;
       \Biggl(\sum_{b}X_b\Biggr)^n
       \prod_{b\in S}X_b.
\eea
The integral of any monomial vanishes
unless every site appears with an even power in the $\phi$' s, all of
which are nonnegative. And the term with $n=0$ and $S=\{\varnothing\}$ equals $1$. As a result,
\bea
Z_\Lambda = 1 + \sum\text{(non-negative contributions)} \geq 1 .
\eea
\end{remark}
\begin{corollary}\label{cor._FE}
Let $Z_{\Lambda}(\kappa^{2},\alpha)$ denote the partition function on the finite lattice $\Lambda$, with $\Lambda_{s}$ sites, and assume that there exist constants 
$C_{\ell},C_{u}\in\mathbb{R}$, independent of $\Lambda_{s}$ and of the lattice spacing $a\in(0,1]$, such that the TUV stability bounds of \eqref{Z-bounds} hold, for all $\Lambda$ and $a$. As in \eqref{FE}, consider the normalized free energy on $\Lambda$
$$
F_{\Lambda}(\kappa^{2},\alpha)
\;=\;
\frac{1}{\Lambda_{s}}\ln Z_{\Lambda}(\kappa^{2},\alpha)\,.
$$
Then, the following statements hold:
\begin{enumerate}
\item[\textnormal{(i)}] For any fixed finite hypercubic lattice $\Lambda$ and all $a\in (0,1]$, we have
\begin{equation}
\label{eq:F-bounds}
C_{\ell}\;\le\;F_{\Lambda}(\kappa^{2},\alpha)\;\le\;C_{u}.
\end{equation}
\item[\textnormal{(ii)}] Let $(\Lambda^{(n)}_s)_{n\in\mathbb N}$ and $(a_n)_{n\in\mathbb N}$ be sequences such that
\bea
\Lambda_{s}^{(n)}\nearrow\infty,
\qquad
a_{n}\searrow 0.
\eea
Then the sequence $\{F_{\Lambda^{(n)}}(\kappa^{2},\alpha)\}_{n\in\mathbb{N}}$ admits a 
convergent subsequence, and every subsequential limit lies in $[C_{\ell},C_{u}]$. 
In particular, the set of infinite volume, continuum limit points of the 
normalized free energy is nonempty and contained in this compact interval.
\end{enumerate}
 \end{corollary}
\begin{proof}
Taking the logarithm in \eqref{Z-bounds} gives
\bea
C_{\ell}\Lambda_{s}
\;\le\;
\ln Z_{\Lambda}(\kappa^{2},\alpha)
\;\le\;
C_{u}\Lambda_{s},
\eea
and division by $\Lambda_{s}>0$ yields \eqref{eq:F-bounds}.
Thus $\{F_{\Lambda}(\kappa^{2},\alpha)\}$ is uniformly bounded for all finite 
$\Lambda$ and all $a>0$. Next, consider any sequence $(\Lambda_{s}^{(n)},a_{n})$ with 
$\Lambda_{s}^{(n)}\!\nearrow\!\infty$ and, subsequently, $a_{n}\!\searrow\!0$ Set 
$F_{n}\doteq F_{\Lambda^{(n)}}(\kappa^{2},\alpha)$.  
By \eqref{eq:F-bounds}, the sequence $\{F_{n}\}$ lies in the closed and bounded 
interval $[C_{\ell},C_{u}]$, which is compact.  
Thus, applying the Bolzano--Weierstrass theorem (see e.g. theorems 2.41 and 2.42 in \cite{Rudin1976principles}), $\{F_{n}\}$ therefore has a convergent 
subsequence with limit in $[C_{\ell},C_{u}]$, proving the claim and resulting in the existence of a finite free normalizeds free-energy in the thermodynamic and continuum limits, at least in the subsequential sense.
\end{proof}
\subsection{The Two-Point Correlation and the Critical Line}

\begin{proof} [Proof of Theorem \ref{thm. 2-pt. funct.}]
Expanding the bond product in \eqref{hybrid-partition} in powers of $\alpha$ gives
\bea \label{prod. aX_b}
\prod_{b}\big(1+\alpha X_b\big)=1+\alpha\,\sum_{b}X_b+\alpha^{2}\sum_{b< b^\prime }X_b X_{b^\prime}+\mathcal{O}(\alpha^3)\,,
\eea
where the notation $b<b'$ indicates that the bond $b$ is not counted twice in the pairing of fields. Equivalently, as the sum runs over spacetime points, only distinct unorded pairs are included.
Consequently, keeping only terms up to $\mathcal{O}(\kappa^{4})$ and $\mathcal{O}(\alpha^{2})$, the partition function becomes
\bea
Z(\kappa^{2},\alpha)=\int d\mu(\phi)\,\Bigg(1+\left(\kappa^{2}+\alpha\right)\sum_{b}X_b-\frac{\alpha^{2}}{2}\sum_{b}X_b^2+\frac{1}{2}\left(\kappa^{2}+\alpha\right)^{2}\sum_{b,\,b^\prime }X_b X_{b^\prime}\Bigg).
\eea
By Wick’s theorem \cite{Gattringer:2010zz,Montvay:1994cy} one immediately has
\bea \label{zero int.}
\int d\mu(\phi)\,X_b=0
\eea
pointwise in $b$. Therefore, expanding $(1/Z)$ as a geometric series and using \eqref{zero int.} to discard the  contributions \ch with only one field $\phi$, we obtain
\bea
\frac{1}{Z(\kappa^2,\alpha)}= 1\,-\,\frac{1}{2}\left(\kappa^{2}+\alpha\right)^{2}\sum_{b,\,b^\prime}\int d\mu(\phi)\,X_b X_b^\prime\,+\,\frac{\alpha^{2}}{2}\sum_{b}\int d\mu(\phi)\,X_b^2\,,
\eea 
In particular, the would–be terms proportional to $\left[-\tfrac12(\kappa^{2}+\alpha)^{2}\right]$ and $\left[(\kappa^{2}+\alpha)^{2}\right]$ that involve a single factor of $X_b$ vanish because of \eqref{zero int.}.

Next, the two–point function, up to the same truncation order in $\kappa^2$ and $\alpha$ as above, can be written as
\bea \label{2pt-G}
G_\Lambda(x,y)=\sum_{i=1}^{4} G_i(x,y),
\eea
where, again by Wick’s theorem,
\bea \label{2-pt. G1}
G_{1}(x,y)=\int d\mu(\phi)\,\,\phi_{x}\phi_{y}=\delta_{xy},
\eea
\bea \label{2-pt. G2}
G_{2}(x,y)=\left(\kappa^{2}+\alpha\right)\sum_{b}\int d\mu(\phi)\,\phi_{x}\phi_{y}X_b=\big(\kappa^{2}+\alpha\big)\sum_{\mu}\Big(\delta_{yx_{\mu}^{+}}+\delta_{yx_{\mu}^{-}}\Big),
\eea
\bea \label{2-pt. G3}
G_{3}(x,y)&=&\frac{1}{2}\left(\kappa^{2}+\alpha\right)^{2}\sum_{b,\,b^\prime}\Bigg\{ \int d\mu(\phi)\,\phi_{x}\phi_{y}X_b X_{b^\prime}-\int d\mu(\phi)\,\phi_{x}\phi_{y}\int d\mu(\phi)\,X_b X_{b^\prime}\Bigg\} \nn\\&=&
\frac{1}{2}\big(\kappa^{2}+\alpha\big)^{2}\sum_{z,\mu}\sum_{z^{\prime},\nu}\Big(\delta_{x[z}\delta_{y\{z^{\prime}}\delta_{z_{\nu}^{\prime+}\}z_{\mu}^{+}]}+\delta_{x[z^{\prime}}\delta_{y\{z}\delta_{z_{\mu}^{+}\}z_{\nu}^{\prime+}]}\Big),
\eea
where, \ch {\bf be more precise!!!! i do not see square brackets above!} for the indices, $[\cdots]$ and $\{\cdots\}$ indicate that the corresponding expression contains the two terms obtained by swapping the relevant indices. And the remaining term of the two-point function \eqref{2pt-G} with $b=b^\prime$ in \eqref{2-pt. G3} yields
\bea \label{2-pt. G4}
G_{4}(x,y)=-\alpha^{2}\sum_{\mu,z}\left(\delta_{xz}\delta_{yz}+\delta_{xz_{\mu}^{+}}\delta_{yz_{\mu}^{+}}\right)=-2d\alpha^{2}\delta_{xy}
\eea
Let $\tilde G(p)$ denote the lattice Fourier transform \footnote{Generally, what is called \emph{momentum space} is crucial for the Fourier representation. 
In precise mathematical terms, this is described by the \emph{Pontryagin dual} of the underlying abelian group \cite{PontrjaginTopGroups}, which is also the setting of the SNAG theorem (a generalization of Stone's theorem) \cite{Folland1994course} that carefully addresses the joint spectrum of momentum operators in higher dimensions. For a corresponding discussion in a lattice setting, see e.g.\ \cite{Naaijkens:2013vqa}. In the continuum, the Pontryagin dual of $\mathbb{R}^d$ is again $\mathbb{R}^d$. 
However, in discrete settings (relevant for our case, though we will not make this explicit) the Pontryagin dual of $\mathbb{Z}^d$ is the $d$-dimensional torus $\mathbb{T}^d$. This is, for example, used explicitly in \cite{Schor:1983aq}.} of $G_\Lambda(x,y)$ and $\tilde{\Gamma}(p)$ its reciprocal. Specifically (up to an overall normalization),
\bea
\tilde{G}(p)=\sum_x G(x)\,\,e^{-ip\cdot x}.
\eea
Evaluating at $p=0$ gives
\bea
\tilde{G}_1(0)=1,\qquad \tilde{G}_2(0)= \kappa^{2}+\alpha,
\eea
while a direct counting of the Wick contractions yields
\bea
\tilde{G}_3(0)=4d^{2}\left(\kappa^{2}+\alpha\right)^{2},\quad\mathrm{and}\quad\tilde{G}_4(0)=-2d\alpha^{2}.
\eea
Summing these contributions, one finds the convolution inverse at zero momentum
\bea
\tilde{\Gamma}(0)=1-2d\left(\kappa^{2}+\alpha\right)+2d\alpha^{2},
\eea
and setting $\tilde{\Gamma}(0)=0$ yields the critical line, as  stated in \eqref{critical}.
\end{proof}
\begin{proof}[Proof of Theorem \ref{thm. correl. funct.}] The result follows from application of a standard polymer-like expansion \cite{Glimm:1987ng,Simon1993}. For instance, we may consider the following interpolated representation of the partition function (\ref{hybrid-partition}) for the hybrid model. 
\bea\label{Interpolation}
Z_{\Lambda}(\kappa^{2},\alpha;\{s\})
=\int d\mu(\phi)\, \exp\!\Big(\kappa^{2}\sum_{b}s_b \, X_b\Big)
\prod_{b^\prime}\big(1+\alpha\,s_{b^\prime} \, X_{b^\prime}\big)\,.
\eea
In a similar way, we also define interpolated unnormalized correlations and, by dividing by the interpolated partition function $Z_{\Lambda}(\kappa^{2},\alpha;\{s\})$, the interpolated normalized correlations. 

Above, in \eqref{Interpolation}, we used a family $\{s\}$ of interpolating parameters $s_{b}\in[0,1]$, associated with each bond $b\in\Lambda$. Applying truncated first order Taylor expansions in each $s_b$, like
$$
Z_{\Lambda}(\kappa^{2},\alpha)\,=\,Z_{\Lambda}(\kappa^{2},\alpha;\{s\}=1)\,=\,\prod_{b\in\Lambda} \left[ I_b\,+\, J_b\right]\:Z_{\Lambda}(\kappa^{2},\alpha;\{s\})\,,
$$
with $\{s\}=1$ meaning any $s_b=1$ and, for a function $f(s_b)$, we have $I_b\,f(s_b)\,=\,f(0)$ and $J_b\,f(s_b)\,=\,\int_0^1\,\left(df/ds_b\right)\!(s_b)\:ds_b$, we obtain an expansion in terms of disjoint polymers. The same procedure is applied to unnormalized correlations such that the expansion is performed both in the numerator and in the denominator of normalized correlations. 

Noting that our Gaussian measure $d\mu(\phi)$ involves only a local covariance, applying standard combinatorial arguments for cluster/polymer expansions lead to a convergent expansion for small enough parameters $\kappa^2$ and $\alpha$, since we are away from the critical line and correlations decay fast enough. Then the claim follows. For completeness we provide a rough sketch yielding exponential decay in appendix \ref{remark_polymer}.
\end{proof}

\subsection{The Non-Vanishing of the Truncated Four-Point Correlation}\label{sec4}
In the proof of the next theorem, we will rely on the explicit expressions for the components of the two-point function derived in the proof of theorem~\ref{thm. 2-pt. funct.}, as these are needed to analyze the truncated contribution to the four-point function.

\begin{proof}[Proof of Theorem \ref{thm. 4-pt. funct.}]
The non-truncated part of the four-point function in \eqref{tr-4pt} is obtained completely analogously to the two-point function in the proof of Theorem~\ref{thm. 2-pt. funct.}. 
For the truncated part, the relevant pairwise products follow directly from \eqref{2-pt. G1}--\eqref{2-pt. G4}. 
Combining these results proves the claim. 
For completeness, we provide the full proof in appendix~\ref{proof 4.1}.
\end{proof}
\section{Conclusions}\label{sec5}
We have introduced a scalar quantum field model on a  hypercubic lattice $\Lambda$ of spacetime dimension $d=4$. Our model obeys O-S positivity with respect to the physical quantum mechanical Hilbert space which we constructed in Section \ref{sec2.1}. Our model has parameters $\kappa^2>0$ (hopping parameter) and $\alpha>0$ (interaction coupling) and is based on the use of scaled fields, as previously introduced by some of us (see \cite{daVeiga:2022twu} and references cited therein). The scaling we use is noncanonical and corresponds to a partial renomalization of the model. We refer to our scalar model as the {\em hybrid model}.

Considering this hybrid model, we prove that, at least for the domain of small parameters ($0<\kappa^2<\kappa^2_0$ and $0<\alpha<\alpha_0$, for some values of $0<\kappa^2_0\ll1$ and $0<\alpha_0\ll 1$) the model satisfies:
\begin{enumerate}
	\item The model partition function admits a thermodynamic and ultraviolet stable stability bounds, leading to a finite normalized free energy in the thermodynamic, and then the continuum limit, at least in the sense of subsequences; 
	\item The corresponding correlations are finite in the thermodynamic limit; 
	\item \ch {\bf ???? CHECK!} the correlations also exist in the continuum limit; 
	\item A computation of the truncated four-point correlation in the first ordersof a Taylor expansion about $\kappa^2$ and $\alpha$, at the origin, shows that the truncated four-point function is non-vanishing.
\end{enumerate}

Thus, the hybrid model is an interacting nontrivial (non-Gaussian) model. As we know, this is the first example of a field theoretical model supporting the compatibility of quantum mechanics and special relativity.

\appendix
\section{More details on the Proof of Theorem \ref{thm. correl. funct.}}
\label{remark_polymer}
Away from the critical line~\eqref{critical}, e.g. as in the small parameter domain we use, the hybrid model Boltzmann factor~\eqref{action} admits a convergent polymer expansion. Using a truncated Taylor expansion of both the $\alpha$ interaction term and the exponential hopping term  in~\eqref{Interpolation} with respect to the interpolation parameters $\{s_b\}$ (cf.~\eqref{Z-expansion}), and regrouping the resulting terms according to their connected bond supports, the interpolated partition function can be rewritten as
\bea \label{parition-polymer}
Z_\Lambda(\kappa^2,\alpha;\{s\})\,=\,\sum_{\mathcal F\ \mathrm{compatible}}
\;\prod_{\gamma\in\mathcal F} w_{\kappa^2,\alpha}(\gamma;\{s\})\,.
\eea
where $\gamma$ runs over finite polymers (connected sets of bonds), $\mathcal F$ ranges over finite families of pairwise compatible polymers, and $w(\gamma)$ denotes the polymer activity; see \cite{Glimm:1987ng,Simon1993}. \ch Here, connectedness of bonds is defined by the existence of bond common lattice endpoints in bonds and compatibility corresponds to disjointness of polymer supports and guarantees factorization of their activities.
	
Let $\mathbb P$ denote the (countable) set of polymers, equipped with the counting measure weighted by the activities\footnote{With $\mathbb P$ as above, one equips it with the discrete $\sigma$--algebra and defines the (complex) measure $\mu$ by $\mu(A)=\sum_{\gamma\in A} w(\gamma)$ for $A\subset\mathbb P$, so that in particular $\mu(\{\gamma\})=w(\gamma)$.}, i.e.\ $\mu(\{\gamma\})=w(\gamma)$. Define the symmetric compatibility function $\zeta:\mathbb P\times\mathbb P\to\{-1,0\}$ by (here $\gamma\sim\gamma'$ means that the polymer supports are disjoint)
\bea
\zeta(\gamma,\gamma')=
\begin{cases}
		-1, & \gamma\not\sim\gamma',\\
		0,  & \gamma\sim\gamma',
\end{cases}
\eea
and let $g(\gamma,\gamma')\ge0$ be a distance weight growing linearly with the lattice graph distance between $\supp(\gamma)$ and $\supp(\gamma')$. \ch Setting $\zeta_g(\gamma,\gamma')\doteq \zeta(\gamma,\gamma')\,e^{g(\gamma,\gamma')}$, and since $(\kappa^2,\alpha)$ lies in the small parameter domain,and hence in  polymer--convergent (non--critical) regime, the activities are sufficiently small such that there exists a nonnegative function $h:\mathbb P\to[0,\infty)$ and $s_*>0$ such that for all choices of interpolation parameters $\{s_b\}$ with $|s_b|\le s_*$ for every bond $b$, the weighted Koteck\'y--Preiss condition holds uniformly, cf. \cite{Kotecky1986},
\bea \label{K-P}
\sup_{\substack{\{s\}:\,|s_b|\le s_*\\ \forall b}}\,\sum_{\gamma'\in\mathbb P}
|w(\gamma';\{s\})|\;|\zeta_g(\gamma,\gamma')|\,e^{h(\gamma')}\;\le\;h(\gamma),
\qquad \forall\,\gamma\in\mathbb P .
\eea

For each fixed choice of interpolation parameters $\{s_b\}$ with $|s_b|\le s_*$ the polymer activities define a valid polymer model satisfying the weighted Koteck\'y--Preiss condition. Since the bound is uniform in $\{s_b\}$, the conclusions of Theorem~3 of~\cite{Ueltschi2005} hold uniformly in this domain. (The validity of \eqref{K-P} uniformly in $\{s_b\}$ may require additional bounds on the activities $w(\gamma;\{s\})$, which are not addressed here.) Now we define the connected correlation function associated with the interpolation parameters by (cf. \cite{Brydges1986})
\bea
G_\Lambda(b,b')\;\doteq\;\left.	\frac{\partial^2}{\partial s_b\,\partial s_{b'}}
\ln Z_\Lambda(\kappa^2,\alpha;\{s\})\right|_{s=0}\,,
\eea
where $Z_\Lambda(\kappa^2,\alpha;\{s\})$ is given in~\eqref{Interpolation}. In the polymer representation, $G_\Lambda(b,b')$ admits an expansion over connected polymers whose supports intersect both bonds $b$ and $b'$. The weighted Koteck\'y--Preiss condition~\eqref{K-P} then implies exponential decay with respect to the bond distance, and hence
\bea \label{exp decay}
|G_\Lambda(x,y)| \;\le\; C\,e^{-m\,\dist(x,y)},
\eea
where $x,y$ are lattice sites incident to the bonds $b,b'$, respectively, and $C>0$ is independent of $x,y$ and of the number of lattice sites $\Lambda_s$.
	
While a complete proof of~\eqref{exp decay} may require additional work, we briefly indicate the underlying mechanism. The restriction to compatible polymer families in~\eqref{parition-polymer} may be implemented by introducing pairwise factors $1+\zeta(\gamma,\gamma')$, which vanish whenever two polymers overlap. In this representation, the partition function involves only these compatibility factors, whereas the logarithm of the partition function is expanded in terms of connected graphs whose edges carry factors $\zeta(\gamma,\gamma')$. Consequently, derivatives of $\ln Z_\Lambda$ with respect to the interpolation parameters isolate contributions from connected polymer configurations linking the corresponding bonds. The Koteck\'y--Preiss bound \eqref{K-P} provides uniform control over the weights of such connected polymer families. In particular, any connected polymer family contributing to the two-point function between two distant bonds must contain polymers whose supports span a distance comparable to their separation. Since the activity of a polymer is obtained as a product of local expansion factors associated with the bonds in its support, which are uniformly small away from criticality, polymer activities decay exponentially with the size of their supports. The bound \eqref{K-P} then ensures that the sum over all connected polymer families linking the two bonds is exponentially suppressed in their separation.
\begin{proposition}
For any $\gamma\in(0,1)$ the absolute value of the $2n$-point correlation \eqref{2n-pt.funct.} satisfies
\bea \label{2n-pt. bound}
\bigl|G_\Lambda(x_1,\ldots,x_{2n})\bigr|\,\le\,	\dfrac{[e(1-\gamma)]^{-n}}{Z_\Lambda}\,	\int\;\prod_{x\in\Lambda}\, \dfrac{d\phi_{x}}{\sqrt{2\pi}}\,\exp\Bigl[-\dfrac{\gamma}{2}\sum_{x}\phi_{x}^{2}
\,+\,(\kappa^{2}+\alpha)\sum_{b}X_b\Bigr].
\eea
\end{proposition}
\begin{proof}
Using \eqref{exp_UB} for each factor in the product over bonds $b$ yields the pointwise bound
\bea \label{B-bound}
B_\Lambda(\phi)\,\le\,\exp\Bigl[(\kappa^{2}+\alpha)\sum_{b}X_b\Bigr]\,.
\eea
	
Inserting \eqref{B-bound} into the definition of $G_\Lambda$ \eqref{2n-pt.funct.} and taking absolute values gives
\bea \label{abs G}
\bigl|G_\Lambda(x_1,\ldots,x_{2n})\bigr|\,\le\,\frac{1}{Z_\Lambda}\,\int\, d\mu(\phi)\;|\phi_{x_1}\cdots\phi_{x_{2n}}|\,	\exp\Bigl[ (\kappa^{2}+\alpha)\sum_{b}X_b\Bigr] \,.
\eea

To remove the field factors downstairs in the integrand, we employ a (convex) decomposition of the Gaussian weight at each site
\bea \label{gamma eq.}
\exp\left\{-\frac{1}{2}\phi_{x}^{2}\right\}\,=\,\exp \left\{-\frac{\gamma}{2}\phi_{x}^{2}\right\}\;\exp\left\{-\frac{1-\gamma}{2}\phi_{x}^{2}\right\}\,,
\qquad \gamma\in(0,1)\,.
\eea

Using \eqref{gamma eq.} only into the numerator of \eqref{abs G} yields
\bea
\bigl| G_\Lambda(x_1,\ldots,x_{2n})\bigr| &\leq& \dfrac{1}{Z_\Lambda}\,\int\; \prod_{x\in\Lambda}\,\dfrac{d\phi_{x}}{\sqrt{2\pi}}\,\left\{ \prod_{i=1}^{2n}|\phi_{x_{i}}|\,\exp \left[-\dfrac{1-\gamma}{2}\phi_{x_{i}}^{2}\right] \right\} 
\vspace{2mm}\nn\\ &&\times\;\exp\!\Bigl[-\frac{\gamma}{2}\sum_{x}\phi_{x}^{2}\,+\,(\kappa^{2}+\alpha)\sum_{b}X_b\Bigr].
\eea

For any $u\in\mathbb{R}$ and $\delta>0$, the function $f(u)=|u|e^{-\frac{\delta}{2}u^{2}}$ attains its maximum at $|u|=\delta^{-1/2}$ (see Remark \ref{phi-bound}), giving
\bea \label{max-bound}
|u|\,e^{-\frac{\delta}{2}u^{2}}\le (e\delta)^{-1/2}.
\eea
	
Setting $\delta=1-\gamma$ and applying \eqref{max-bound} to each insertion point $x_{i}$ produces the uniform bound
\bea \label{phi2n-bound}
\prod_{i=1}^{2n}|\phi_{x_{i}}|\,
\exp\left\{-\frac{1-\gamma}{2}\phi(x_{i})^{2}\right\}\,\leq\, [e(1-\gamma)]^{-n}\,.
\eea
Since the right-hand side is independent of the fields, it factors out of the integral, immediately yielding \eqref{2n-pt. bound}.
\end{proof}

\begin{remark}
Without the decomposition \eqref{gamma eq.}, the right-hand side of \eqref{abs G} would require evaluating the Gaussian correlation involving $2n$ fields, which by Wick's theorem expands into $(2n-1)!!$ pairings and grows factorially with $n$.	

Inequality \eqref{2n-pt. bound} replaces this combinatorial sum by the partition function of a \emph{free} theory with parameter $\gamma$ and coupling $\kappa^{2}+\alpha$, multiplied by the explicit exponential prefactor $[e(1-\gamma)]^{-n}$. Consequently, the dependence on the order $n$ is reduced from factorial to purely exponential, rendering the bound uniform in the insertion points $\{x_{i}\}$ and explicit for large $n$.
\end{remark}

\begin{remark} \label{phi-bound}
To prove \eqref{phi2n-bound} it suffices to bound a single factor. Consider the function $f(u)=|u|\,e^{-\frac{\delta}{2}u^{2}}$ with $\delta=1-\gamma>0$. Since $f$ is even, we may restrict to $u\geq 0$ and maximize $f(u)=u\,e^{-\frac{\delta}{2}u^{2}}$. Differentiating,
\bea
f'(u)=e^{-\frac{\delta}{2}u^{2}}-\delta u^{2}e^{-\frac{\delta}{2}u^{2}}
		=e^{-\frac{\delta}{2}u^{2}}\bigl(1-\delta u^{2}\bigr)\,.
\eea
	
Setting $f'(u)=0$ yields the unique positive critical point $u_{c}=\delta^{-1/2}$. Evaluating $f$ at this point,
\bea
f(u_{c})=\delta^{-1/2}e^{-1/2}\,.
\eea

Since $f(0)=0$ and $\lim_{u\to\infty}f(u)=0$, this is the global maximum. Thus for each $i=1,\dots,2n$,
\bea
|\phi_{x_{i}}|\,e^{-\frac{1-\gamma}{2}\phi_{x_{i}}^{2}}\,\le \bigl[e(1-\gamma)\bigr]^{-1/2}\,.
\eea

Taking the product over all $2n$ insertion points gives the desired bound.
\end{remark}

\begin{corollary} \label{cor. 2-pt funct. 2}
	Let $(\kappa^{2},\alpha)$ lie away from the critical line \eqref{critical}, as is the case of the small parameter domain. Then, \ch if $d=3,4,\ldots$, the two point correlation admits a well-defined upper-bound momentum (spectral) representation in the thermodynamic limit, including coincident-point contributions. Moreover, the two point function remains uniformly bounded in the continuum limit $a \searrow 0$.
\end{corollary}
\begin{proof}
	Using \eqref{exp_UB}, the two--point function admits an upper bound such that the partition function coincides with that of the free scaled scalar field model. Then, using \eqref{op. B}, one immediately finds the symmetric matrix
\bea
M(x,y)=\dfrac{\kappa^{2}+\alpha}{2}\;\sum_\mu\big(\delta_{yx^+_\mu}+\delta_{yx^-_\mu}\big)-\frac{1}{2}\,\delta_{xy}\,.
\eea

The orthonormal eigenvectors $f_p(x)=|\Lambda_s|^{-1/2}e^{ip\cdot x}$ along with the quadratic form $(\phi,M\phi)$, cf. \eqref{op. B}, for the partition function \eqref{hybrid-partition} (after using \eqref{exp_UB}) yields
\bea
\lambda(p)=\frac12\Big(2\,(\kappa^{2}+\alpha)\sum_{\mu}\cos(ap^\mu)-1\Big)\,.
\eea
Consequently, in the thermodynamic limit $\Lambda_s\nearrow\infty$, the two--point function admits the momentum--space representation (see \cite{FariadaVeiga:2022bqi})---as an upper bound \ch{\bf check denominator parenthesis in cos!}
\bea
G_\Lambda (x,y)\,=\,\frac{1}{2(2\pi)^{d}}\;\int_{(-\pi,\pi]^{d}}d^{d}q\,
\dfrac{\exp\left[i\,\dfrac{q\cdot(x-y)}{a}\right]}{\frac{1}{2}\left[\kappa^{2}+\alpha)\,\sum_{\mu=0,\ldots,(d-1)}\,\left(2\cos q^{\mu}-1\right)\right]}\,,
\eea
where the assumption that $(\kappa^{2},\alpha)$ is in the small parameter domain guarantees that the denominator is uniformly bounded away from zero, especially for the contributions at the coincident points, unless $d>2$.
	
The exponential decay \eqref{exp decay} guarantees integrability of the spectral representation \ch for $x\not- y$, and integrability also holds for coincident point contributions. The analysis of the thermodynamic and continuum limits then follows exactly as in appendix~A.2 of \cite{daVeiga:2022twu}, the integral defining $G_\Lambda(x,x)$ is uniformly bounded for $a\in(0,1]$, and the (Lebesgue) dominated convergence yields a finite limit as $a\searrow 0$.
\end{proof}

\section{Expanded Proof of Theorem \ref{thm. 4-pt. funct.}} \label{proof 4.1}

\ch Only in this appendix, we adopt the following notations. Let $X = \{x_1, x_2, x_3, x_4\}$, $Z = \{z, z^+_\mu\}$, and $Z^\prime= \{z^\prime, z^{\prime +}_\mu\}$. Let $G_\Lambda(X)$, or its components $G_i(X)$ (cf.\ the proof of Theorem~\ref{thm. 2-pt. funct.}), denote the integral with respect to the measure~\eqref{measure}, whose integrand contains a product of four fields, like
\bea \label{4phi}
O(X)\doteq\phi_{x_1}\phi_{x_2}\phi_{x_3}\phi_{x_4}\,.
\eea
Specifically,
\bea
G_\Lambda(X)\doteq\int d\mu(\phi)\; O(X),\qquad G_\Lambda(X,Z)\doteq \sum_b\int d\mu(\phi) \:O(X) X_b\,,
\eea
and so on.
Then, using Wick's theorem, one directly computes the components 
(the terms corresponding to the various coefficients, as in the proof of 
Theorem~\ref{thm. 2-pt. funct.}) of the four-point function $G(x)$ as follows.
\bea
G_{1}(X)=\delta_{x_{1}x_{2}}\delta_{x_{3}x_{4}}+\delta_{x_{1}x_{3}}\delta_{x_{2}x_{4}}+\delta_{x_{1}x_{4}}\delta_{x_{2}x_{3}}.
\eea
\bea
G_{2}(X,Z)=&\Big(\kappa^{2}+\alpha\Big)\sum_{\mu}\Big\{ 2\delta_{x_{1}x_{2}}\delta_{x_{3}\left(x_{4}\right)_{\mu}^{-}}+2\delta_{x_{1}x_{3}}\delta_{x_{2}\left(x_{4}\right)_{\mu}^{-}}+2\delta_{x_{1}x_{4}}\delta_{x_{2}\left(x_{3}\right)_{\mu}^{-}}\nn\\&\qquad\quad\quad\:\:\:+\,2\delta_{x_{1}\left(x_{2}\right)_{\mu}^{-}}\delta_{x_{3}x_{4}}+2\delta_{x_{1}\left(x_{3}\right)_{\mu}^{-}}\delta_{x_{2}x_{4}}+2\delta_{x_{1}\left(x_{4}\right)_{\mu}^{-}}\delta_{x_{2}x_{3}}\Big\}\,,
\eea
\ch where, $\mu=0,1,\ldots,(d-1)$. The third term is slightly more involved; in particular, with the notation introduced above, it can be written as
\bea \label{4pt-G3}
G_3(X,Z,Z^\prime)=\frac{1}{2}\left(\kappa^{2}+\alpha\right)^{2}\sum_{b,\,b^\prime} \Bigg\{\int d\mu(\phi) \;O(X) X_b X_{b^\prime}-\int d\mu(\phi)\; O(X) \int d\mu(\phi) X_b X_{b^\prime}\Bigg\}\,,
\eea
which, upon a delicate computation, gives the following. 
The first term produces $7!!$ total terms, among which 78 yield nonzero 
contributions after using $\delta _{zz^+_\mu} = \delta_{z^\prime z^{\prime +}_\mu} = 0$. This further reduces to $72$ after subtracting the nonzero contributions coming from the second term in \eqref{4pt-G3}. The result can be written compactly as follows.
\bea \label{4pt-G3_2}
G_3(X,Z,Z^\prime)&=&\frac{1}{2}\left(\kappa^{2}+\alpha\right)^{2}\,\sum_{b,\,b^\prime}\Bigg\{\delta_{x_{1}x_{2}}\left(\delta_{x_{3}[z}\delta_{x_{4}\{z^{\prime}}\delta_{z_{\nu}^{\prime+}\}z_{\mu}^{+}]}+\delta_{x_{3}[z^{\prime}}\delta_{x_{4}\{z}\delta_{z_{\mu}^{+}\}z_{\nu}^{\prime+}]}\right)\nn\\&&+\,\delta_{x_{1}x_{3}}\left(\delta_{x_{2}[z}\delta_{x_{4}\{z^{\prime}}\delta_{z_{\nu}^{\prime+}\}z_{\mu}^{+}]}+\delta_{x_{2}[z^{\prime}}\delta_{x_{4}\{z}\delta_{z_{\mu}^{+}\}z_{\nu}^{\prime+}]}\right)\nn\\&&+\,\delta_{x_{1}x_{4}}\left(\delta_{x_{2}[z}\delta_{x_{3}\{z^{\prime}}\delta_{z_{\nu}^{\prime+}\}z_{\mu}^{+}]}+\delta_{x_{2}[z^{\prime}}\delta_{x_{3}\{z}\delta_{z_{\mu}^{+}\}z_{\nu}^{\prime+}]}\right)\nn\\&&+\,\delta_{x_{1}z}\left(\delta_{x_{2}[x_{3}}\delta_{x_{4}]\{z^{\prime}}\delta_{z_{\nu}^{\prime+}\}z_{\mu}^{+}}+\delta_{x_{2}z_{\mu}^{+}}\delta_{x_{3}\{z^{\prime}}\delta_{z_{\nu}^{\prime+}\}x_{4}}+\delta_{x_{2}\{z^{\prime}}\delta_{x_{3}(x_{4}}\delta_{z_{\mu}^{+}z_{\nu}^{\prime+}\})}\right)\nn\\&&+\,\delta_{x_{1}z_{\mu}^{+}}\left(\delta_{x_{2}[x_{3}}\delta_{x_{4}]\{z^{\prime}}\delta_{z_{\nu}^{\prime+}\}z}+\delta_{x_{2}z}\delta_{x_{3}\{z^{\prime}}\delta_{z_{\nu}^{\prime+}\}x_{4}}+\delta_{x_{2}\{z^{\prime}}\delta_{x_{3}(x_{4}}\delta_{zz_{\nu}^{\prime+}\})}\right)\nn\\&&+\,\delta_{x_{1}z^{\prime}}\left(\delta_{x_{2}[x_{3}}\delta_{x_{4}]\{z}\delta_{z_{\mu}^{+}\}z_{\nu}^{\prime+}}+\delta_{x_{2}z_{\nu}^{\prime+}}\delta_{x_{3}\{z}\delta_{z_{\mu}^{+}\}x_{4}}+\delta_{x_{2}\{z}\delta_{x_{3}(x_{4}}\delta_{z_{\mu}^{+}\}z_{\nu}^{\prime+})}\right)\nn\\&&+\,\delta_{x_{1}z_{\nu}^{\prime+}}\left(\delta_{x_{2}[x_{3}}\delta_{x_{4}]\{z}\delta_{z_{\mu}^{+}\}z^{\prime}}+\delta_{x_{2}z^{\prime}}\delta_{x_{3}\{z}\delta_{z_{\mu}^{+}\}x_{4}}+\delta_{x_{2}\{z}\delta_{x_{3}(x_{4}}\delta_{z_{\mu}^{+}\}z^{\prime})}\right)\Bigg\}\,,
\eea
where $[\cdots]$ and $\{\cdots\}$ carry the same meaning as in \eqref{2-pt. G3}, while $(\cdots)$ denotes the cyclic permutations of the three indices, 
yielding three terms from each expression. Finally, the fourth term is computed analogously, with $z = z'$ and 
$\mu = \nu$ in \eqref{4pt-G3_2}, yielding
\bea
G_4(X,Z,Z^\prime)&=&-\frac{\alpha^{2}}{2}\sum_{\mu}\Big(4\delta_{x_{1}x_{2}}\delta_{x_{3}x_{4}}+4\delta_{x_{1}x_{3}}\delta_{x_{2}x_{4}}+4\delta_{x_{1}x_{4}}\delta_{x_{2}x_{3}}+6\delta_{x_{1}x_{2}}\delta_{\left(x_{3}\right)_{\mu}^{-}\left(x_{4}\right)_{\mu}^{-}}\nn\\&&+\,6\,\delta_{x_{1}x_{3}}\delta_{\left(x_{2}\right)_{\mu}^{-}\left(x_{4}\right)_{\mu}^{-}}+6\delta_{x_{1}x_{4}\delta_{\left(x_{2}\right)_{\mu}^{-}\left(x_{3}\right)_{\mu}^{-}}}+6\delta_{\left(x_{1}\right)_{\mu}^{-}\left(x_{4}\right)_{\mu}^{-}}\delta_{x_{2}x_{3}}\nn\\&&+\,6\,\delta_{\left(x_{1}\right)_{\mu}^{-}\left(x_{3}\right)_{\mu}^{-}}\delta_{x_{2}x_{4}}+6\delta_{\left(x_{1}\right)_{\mu}^{-}\left(x_{2}\right)_{\mu}^{-}}\delta_{x_{3}x_{4}}\Big),
\eea
where the sum over $z$ has already been taken care of.

Now, the truncated part of the four-point function \eqref{tr-4pt} 
contains three terms, each of which is a product of two two-point functions. The two-point function in question is given directly by \eqref{2pt-G}, together with the relevant terms appearing in its proof, with $x \to x_1$ and $y \to x_2$. Then straightforward algebraic manipulations yield
\bea \label{prod. 2-pt}
\!\!\!\!G_\Lambda(x_1,x_2)\; G_\Lambda(x_3,x_4)\!\! &=&\!\!\delta_{x_{1}x_{2}} \delta_{x_{3}x_{4}}+\left(\kappa^{2}+\alpha\right)\sum_{\mu}\left(2\delta_{x_{1}x_{2}}\delta_{x_{3}\left(x_{4}\right)_{\mu}^{-}}+2\delta_{x_{1}\left(x_{2}\right)_{\mu}^{-}}\delta_{x_{3}x_{4}}\right)\nn\\&&\!\!+\,\frac{1}{2} \left(\kappa^{2}+\alpha\right)^{2}\Bigg(\delta_{x_{1}x_{2}}\left(\delta_{x_{3}[z}\delta_{x_{4}\{z^{\prime}}\delta_{z_{\nu}^{\prime+}\}z_{\mu}^{+}]}+\delta_{x_{3}[z^{\prime}}\delta_{x_{4}\{z}\delta_{z_{\mu}^{+}\}z_{\nu}^{\prime+}]}\right)\nn\\&&\!\!+\,\delta_{x_{3}x_{4}}\left(\delta_{x_{1}[z}\delta_{x_{2}\{z^{\prime}}\delta_{z_{\nu}^{\prime+}\}z_{\mu}^{+}]}+\delta_{x_{1}[z^{\prime}}\delta_{x_{2}\{z}\delta_{z_{\mu}^{+}\}z_{\nu}^{\prime+}]}\right)\Bigg)\nn\\
&&\!\!-\,\frac{\alpha^{2}}{2}\sum_{\mu}\left(4\delta_{x_{1}x_{2}}\delta_{x_{3}x_{4}}+2\delta_{x_{1}x_{2}}\delta_{\left(x_{3}\right)_{\mu}^{-}\left(x_{4}\right)_{\mu}^{-}}+2\delta_{\left(x_{1}\right)_{\mu}^{-}\left(x_{2}\right)_{\mu}^{-}}\delta_{x_{3}x_{4}}\right)\,\!.
\eea
The remaining two terms are then obtained by the direct observation in \eqref{prod. 2-pt}, with $x_2 \leftrightarrow x_3$ and $x_2\leftrightarrow x_4$, respectively. Thus, combining everything, simplifying the resulting expressions, and summing over $z$ and $z^\prime$, we finally obtain
\bea
G^{\mathrm{T}}(X)&=&-4d\alpha^{2}\Big(\delta_{x_{1}x_{2}}\delta_{x_{3}x_{4}}+\delta_{x_{1}x_{3}}\delta_{x_{2}x_{4}}+\delta_{x_{1}x_{4}}\delta_{x_{2}x_{3}}\Big)\nn\\&&+\,2\left(\kappa^{2}+\alpha\right)^{2}\sum_{\mu,\nu}\Big\{ \delta_{x_{1}\left(x_{2}\right)_{\mu}^{-}}\delta_{x_{3}\left(x_{4}\right)_{\nu}^{-}}+\delta_{x_{1}\left(x_{3}\right)_{\mu}^{-}}\delta_{x_{2}\left(x_{4}\right)_{\nu}^{-}}\nn\\&&+\,\delta_{x_{1}\left(x_{4}\right)_{\mu}^{-}}\delta_{x_{2}\left(x_{3}\right)_{\nu}^{-}}+\delta_{x_{1}\left(x_{2}\right)_{\nu}^{-}}\delta_{x_{3}\left(x_{4}\right)_{\mu}^{-}}+\delta_{x_{1}\left(x_{3}\right)_{\nu}^{-}}\delta_{x_{2}\left(x_{4}\right)_{\mu}^{-}}\nn\\&&+\,\delta_{x_{1}\left(x_{4}\right)_{\nu}^{-}}\delta_{x_{2}\left(x_{3}\right)_{\mu}^{-}}\Big\}.
\eea

\bibliographystyle{utcaps.bst}
\bibliography{Ref.bib}

\end{document}